\documentclass[letterpaper,11pt]{article}

\usepackage{amsthm,amsmath,amssymb}
\usepackage{graphicx}
\usepackage[small]{caption}

\newtheorem{theorem}{Theorem}
\newtheorem{lemma}{Lemma}

\def\etal{{et~al.}}
\def\ie{{i.e.}}
\def\eg{{e.g.}}

\def\F{{\mathcal F}}
\def\L{{\mathcal L}}
\def\S{{\mathcal S}}
\def\T{{\mathcal T}}

\newcommand{\eps}{\varepsilon}

\newcommand{\old}[1]{}
\newcommand{\later}[1]{}

\begin{document}

\title{On the approximability of covering points by lines\\and related
  problems\thanks{%
Supported in part by NSF grant DMS-1001667 awarded to the first author.
The work of the second author was performed during his sabbatical leave
in Fall 2012.}}

\author{Adrian Dumitrescu\thanks{%
Department of Computer Science,
University of Wisconsin--Milwaukee, USA\@.
Email:~\texttt{dumitres@uwm.edu}.}
\and
Minghui Jiang\thanks{%
Department of Computer Science,
Utah State University,
Logan, USA\@.
Email: \texttt{mjiang@cc.usu.edu}.}}

\maketitle
\thispagestyle{empty}

\begin{abstract}
Given a set $P$ of $n$ points in the plane, {\sc Covering Points by Lines}
is the problem of finding a minimum-cardinality set $\L$ of lines
such that every point $p \in P$ is incident to some line $\ell \in \L$.
As a geometric variant of {\sc Set Cover}, {\sc Covering Points by Lines} 
is still NP-hard. 
Moreover, it has been proved to be APX-hard, and hence
does not admit any polynomial-time approximation scheme unless P $=$ NP\@.
In contrast to the small constant approximation lower bound
implied by APX-hardness,
the current best approximation ratio for {\sc Covering Points by Lines} 
is still $O(\log n)$, namely the ratio achieved by the greedy
algorithm for {\sc Set Cover}.

In this paper,
we give a lower bound of $\Omega(\log n)$
on the approximation ratio of the greedy algorithm for {\sc Covering Points by Lines}.
We also study several related problems including
{\sc Maximum Point Coverage by Lines},
{\sc Minimum-Link Covering Tour},
{\sc Minimum-Link Spanning Tour},
and {\sc Min-Max-Turn Hamiltonian Tour}.
We show that all these problems are either APX-hard or at least NP-hard.
In particular, our proof of APX-hardness of {\sc Min-Max-Turn Hamiltonian Tour}
sheds light on the difficulty of {\sc Bounded-Turn-Minimum-Length Hamiltonian Tour},
a problem proposed by Aggarwal et~al.\ at SODA 1997.

\later{
}

\end{abstract}

\newpage
\setcounter{page}{1}
\setcounter{footnote}{0}

\section{Introduction}

Given a set $U$ of $n$ elements and a family $\F$ of $m$ subsets of $U$,
{\sc Set Cover} is the problem of finding a minimum-cardinality subfamily
$\F' \subseteq \F$ whose union is $U$.
It is well-known that
{\sc Set Cover} can be approximated within
$H_n \le \ln n + 1$~\cite{Jo74,Lo75,Ch79},
where $H_n$ is the $n$th harmonic number,
by a simple greedy algorithm that
repeatedly selects a set that covers the most remaining elements;
a more refined analysis~\cite{Sl97} shows that the approximation ratio
of the greedy algorithm is in fact $\ln n - \ln\ln n + \Theta(1)$.
On the other hand, {\sc Set Cover} 
cannot be approximated in polynomial time
within $c \log n$ for some constant $c > 0$
unless P $=$ NP~\cite{RS97},
and
within $(1-\eps)\ln n$
for any $\eps > 0$
unless NP $\subset$ TIME$(n^{O(\log\log n)})$~\cite{Fe98};
see also~\cite{AMS06,LY94}.

The first problem that we study in this paper
is a geometric variant of {\sc Set Cover}.
Given a set $P$ of $n$ points in the plane,
{\sc Covering Points by Lines}
is the problem of finding a minimum-cardinality set $\L$ of lines
such that every point $p \in P$ is in some line $\ell \in \L$.
(Without loss of generality,
we can assume that $n \ge 2$ and that the lines in $\L$ are selected from
the set $\L_P$ of at most $n \choose 2$ lines
with at least two points of $P$ in each line.)

As a restricted version of {\sc Set Cover},
{\sc Covering Points by Lines} may appear as a much easier problem.
Indeed, in terms of parameterized complexity, {\sc Set Cover} is clearly
W[2]-hard when the parameter is the number $k$ of sets in the solution
(as easily seen by a reduction from the canonical W[2]-hard problem
$k$-{\sc Dominating Set}),
while {\sc Covering Points by Lines} 
admits very simple FPT algorithms based on standard
techniques in parameterized complexity such as bounded search tree
and kernelization~\cite{LM05}; see also~\cite{GL06,WLC10}.
In terms of approximability, however,
the current best approximation ratio for {\sc Covering Points by Lines} is still
the same $O(\log n)$ upper bound for {\sc Set Cover}
achieved by the greedy algorithm.
No matching lower bounds are known for {\sc Covering Points by Lines},
although it has been proved to be NP-hard~\cite{MT82}
and even APX-hard~\cite{BHN01,KAR00};
the APX-hardness of the problem implies a constant lower bound
on the approximation ratio, in particular, the problem
does not admit any polynomial-time approximation scheme unless P $=$ NP.

We first give an asymptotically tight lower bound
on the approximation ratio of the greedy algorithm:

\begin{theorem}\label{thm:greedy}
The approximation ratio of the greedy algorithm for {\sc Covering Points by Lines} is
$\Omega(\log n)$.
\end{theorem}

We also prove that {\sc Covering Points by Lines} is APX-hard, unaware\footnote{We
  thank an anonymous source for bringing this to our attention.} of
the previous APX-hardness results of Brod\'en~\etal~\cite{BHN01} and
Kumar~\etal~\cite{KAR00}:
%
\begin{theorem}[Brod\'en \etal~\cite{BHN01} and Kumar \etal~\cite{KAR00}]
\label{thm:min}
{\sc Covering Points by Lines} is APX-hard~\textup{\cite{BHN01,KAR00}}.
This holds even if no four of the given points are collinear~\textup{\cite{BHN01}}.
\end{theorem}

A problem closely related to {\sc Set Cover} is the following.
Given a set $U$ of $n$ elements, a family $\F$ of $m$ subsets of $U$,
and a number $k$,
{\sc Maximum Coverage} is the problem of finding a subfamily $\F' \subseteq \F$
of $k$ subsets whose union has the maximum cardinality.
In the setting of {\sc Covering Points by Lines},
given a set $P$ of $n$ points in the plane and a number $k$,
{\sc Maximum Point Coverage by Lines}
is the problem of finding $k$ lines
that cover the maximum number of points in $P$.
For the general {\sc Maximum Coverage} problem,
the greedy algorithm that repeatedly selects a set that covers the most
remaining elements achieves an approximation ratio of
$1 - 1/e = 0.632\ldots$~\cite[Section~3.9]{Ho97};
this is also the current best approximation ratio for 
{\sc Maximum Point Coverage by Lines}. 
On the other hand, {\sc Maximum Coverage} cannot be approximated better than
$1 - 1/e + \eps$ for any $\eps > 0$ unless P $=$ NP~\cite{Fe98},
while {\sc Maximum Point Coverage by Lines} is only known to be NP-hard as implied by
the NP-hardness of {\sc Covering Points by Lines}~\cite{MT82}.
We show that {\sc Maximum Point Coverage by Lines} is APX-hard too:

\begin{theorem}\label{thm:max}
{\sc Maximum Point Coverage by Lines} is APX-hard.
This holds even if no four of the given points are collinear.
\end{theorem}

Our proof of Theorem~\ref{thm:max} is based on the same construction
as in our proof of Theorem~\ref{thm:min}.
In retrospect, we note that
the construction in our proof is the exact dual of the construction
in~\cite{BHN01}: we cover points by lines; they cover lines by points.
For completeness, we include our proofs of Theorems~\ref{thm:min} and~\ref{thm:max}
in the appendix.

Instead of using lines,
we can cover the points using a polygonal chain of line segments.
Given a set $P$ of $n$ points in the plane,
a \emph{covering tour} is a closed chain of segments
that cover all $n$ points in $P$,
and a \emph{spanning tour} is a covering tour in which
the endpoints of all segments are points in $P$.
The problem {\sc Minimum-Link Covering Tour}
(respectively,
{\sc Minimum-Link Spanning Tour})
aims at finding a covering tour
(respectively,
\emph{spanning tour})
with the minimum number of links (\ie, segments).
Arkin~\etal~\cite{AMP03} proved that {\sc Minimum-Link Covering Tour}
is NP-hard; see also~\cite{ABDFMS05,Ji12} and the references therein.
Strengthening this result, our following theorem shows that
{\sc Minimum-Link Covering Tour} is in fact APX-hard:

\begin{theorem}\label{thm:covering}
{\sc Minimum-Link Covering Tour} is APX-hard.
This holds even if no four of the given points are collinear.
\end{theorem}

We also show that {\sc Minimum-Link Spanning Tour} is NP-hard:

\begin{theorem}\label{thm:spanning}
{\sc Minimum-Link Spanning Tour} is NP-hard.
This holds even if no four of the given points are collinear.
\end{theorem}

Given a set $P$ of $n$ points in the plane,
a \emph{Hamiltonian tour} is a closed polygonal chain of
\emph{exactly} $n$ segments whose $n$ endpoints along the chain are
a circular permutation of the $n$ points in $P$.
Note that every Hamiltonian tour is a spanning tour, but not vice versa,
although every spanning tour can be transformed into a Hamiltonian tour
by subdividing some segments into chains of shorter collinear segments.
\begin{figure}[htb]
\centering\includegraphics[scale=0.67]{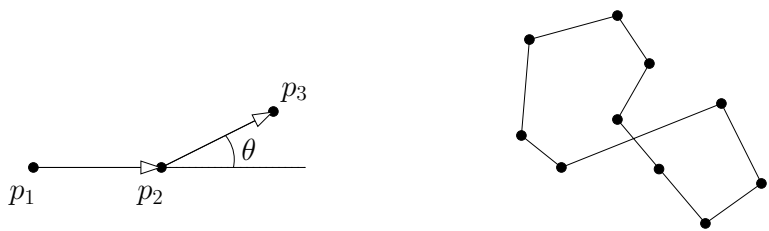}
\caption{%
Left: turning angle at $p_2$.
Right: an obtuse tour of $10$ points.}
\label{fig:f1}
\end{figure}

When three points $p_1,p_2,p_3$ are traversed in this order in a
Hamiltonian tour, 
the \emph{turning angle} at $p_2$, denoted by $\mathrm{turn}(p_1,p_2,p_3)$,
is equal to $\pi - \angle p_1p_2p_3$, where $\angle p_1p_2p_3 \in [0,\pi]$;
see Figure~\ref{fig:f1}. Note that the turning angle belongs to $[0,\pi]$, 
regardless of the direction of the turn (left or right). 
A~tour or path with each turning angle in $[0,\pi/2]$ is called \emph{obtuse}.

In the {\sc Euclidean Traveling Salesman Problem} (ETSP),
given a set $P$ of $n$ points in the plane,
one seeks a shortest Hamiltonian tour that visits each point.
However, frequently other parameters are of interest,
such as in motion planning, where small turning angles are desired.
For example, an aircraft or a boat moving at high speed,
required to pass through a set of given locations,
cannot make sharp turns in its
motion\cite{ACKMS99,ART95,BCL94,Fr89,JC89,LFF07}. 
A rough approximation is provided by paths or tours that are obtuse.
However not all point sets admit obtuse tours or even obtuse paths.
For instance, some point sets require turning angles
at least $5\pi/6$ in any Hamiltonian path~\cite{FW97}.
Moreover, certain point sets (\eg, collinear)
require the maximum turning angle possible, namely $\pi$,
in any Hamiltonian tour.

Aggarwal~\etal~\cite{ACKMS99} have studied the following variant of angle-TSP,
which we refer to as {\sc Min-Sum-Turn Hamiltonian Tour}:
Given $n$ points in the plane,
compute a Hamiltonian tour of the points
that minimizes the total turning angle.
The total turning angle of a tour is the sum of the turning
angles at each of the $n$ points.
They proved that this problem is NP-hard and gave a polynomial-time algorithm
with approximation ratio $O(\log n)$.
They also suggested another natural variant of the basic angle-TSP problem,
where the maximum turning angle at a vertex is bounded and the goal is
to minimize the length measure.

Here we study the computational complexity
of the following two variants of the angle-TSP problem.
The first variant naturally presents itself, however it does not
appear to have been previously studied.
The second variant is one of the two proposed by Aggarwal~\etal~\cite{ACKMS99}.
\begin{itemize} \itemsep -1pt
\item[(I)]
{\sc Min-Max-Turn Hamiltonian Tour}:
Given $n$ points in the plane,
compute a Hamiltonian tour that minimizes the maximum turning angle.
\item[(II)]
{\sc Bounded-Turn-Minimum-Length Hamiltonian Tour}:
Given $n$ points in the plane and an angle $\alpha \in [0,\pi]$,
compute a Hamiltonian tour with each turning angle at most $\alpha$,
if it exists, that has the minimum length.
\end{itemize}

We have the following two results for the two variants of angle-TSP:

\begin{theorem}\label{thm:hamiltonian}
{\sc Min-Max-Turn Hamiltonian Tour} is APX-hard.
\end{theorem}

\begin{theorem}\label{thm:hamiltonian2}
{\sc Bounded-Turn-Minimum-Length Hamiltonian Tour} is NP-hard.
\end{theorem}

\section{An $\Omega(\log n)$ lower bound on the approximation ratio of
  the greedy algorithm} 

In this section we prove Theorem~\ref{thm:greedy}.
Our construction is inspired by a construction of Brimkov~\etal~\cite{BLWM12}
for the related problem {\sc Covering Segments by Points},
which is in turn inspired by a classic lower bound construction for {\sc Vertex Cover}.
This construction shows that there exist graphs with $n$ vertices 
on which the greedy algorithm for {\sc Vertex Cover} achieves
a ratio of $\Omega(\log {n})$. 
\begin{figure}[htb]
\centering\includegraphics[scale=1.05]{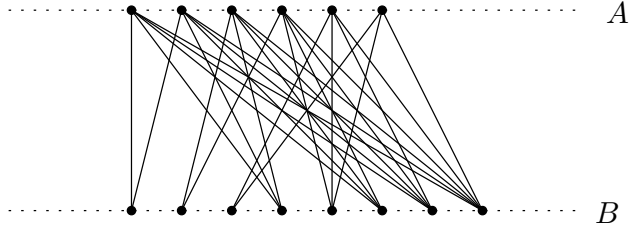}
\caption{A lower bound construction for {\sc Vertex Cover} and 
{\sc Covering Segments by Points} 
($k=6$ in this example).}
\label{fig:greedy1}
\end{figure}

Let $G_k = (A\cup B, E)$ be a bipartite graph
where $A$ is a set of $k$ vertices,
$B$ is a set of $\sum_{i=2}^k \lfloor \frac{k}{i} \rfloor$ vertices
partitioned into $k-1$ subsets $B_2, B_3, \ldots, B_k$,
and $E$ is a set of $\sum_{i=2}^k \lfloor \frac{k}{i} \rfloor \cdot i$ edges.
For $2 \le i \le k$,
each subset $B_i$ contains $\lfloor \frac{k}{i} \rfloor$ vertices
which are connected to $\lfloor \frac{k}{i} \rfloor \cdot i$ vertices in $A$,
with each vertex in $B_i$ connected to exactly $i$ distinct vertices in $A$.
Refer to Figure~\ref{fig:greedy1} for an illustration of the graph $G_k$ with
$k=6$.

Execute the greedy algorithm for {\sc Vertex Cover} on the bipartite graph $G_k$.
In each step of the algorithm,
after a vertex of the maximum degree is selected,
the vertex and its incident edges are removed from the graph.
The crucial observation here is that before each selection,
the degree of each vertex in $A$ is at most the number of subsets $B_i$
that are not empty,
while the degree of each vertex in a non-empty subset $B_i$ is exactly $i$.
Thus the vertex of maximum degree selected in each step
is always from a non-empty subset $B_i$ with the maximum index $i$.
A simple induction shows that the greedy algorithm always selects vertices
from $B_k, B_{k-1}, \ldots, B_2$, in this order, and stops when all vertices
in $B$ are selected.
On the other hand, the set of vertices in $A$ clearly covers all edges too.
Thus the approximation ratio of the greedy algorithm for Vertex Cover is
at least
\begin{align*}
\frac{|B|}{|A|}
= \frac{\sum_{i=2}^k \lfloor \frac{k}{i} \rfloor}{k}
&= \frac{(\sum_{i=1}^k \lfloor \frac{k}{i} \rfloor)- k}{k}
\ge \frac{(\sum_{i=1}^k \frac{k}{i})- 2 k}{k}\\
&= \sum_{i=1}^k \frac{1}{i} - 2
\ge \int_1^{k+1} \frac{1}{x} \mathrm{d}x - 2
= \ln(k+1) - 2,
\end{align*}
which is $\Omega(\log{n})$, where $n=\Theta(k \log{k})$ is the number
of vertices of $G_k$.

We now relate {\sc Vertex Cover} to a geometric problem, 
{\sc Covering Segments by Points}:
Given a set $\S$ of $n$ line segments in the plane, find a set $P$ of
points of minimum size such that each segment in $\S$ contains at least one point in $P$. 
To adapt the construction for {\sc Vertex Cover} to {\sc Covering Segments by Points},
Brimkov~\etal~\cite{BLWM12} place the vertices in $A$ and $B$ in two parallel
lines, with unit distance between consecutive vertices in each line,
and with the vertices in each subset $B_i$ placed consecutively,
as illustrated in Figure~\ref{fig:greedy1}.
Each edge in $G_k$ corresponds to
a line segment in $\S$ with the two vertices as the endpoints.
Without loss of generality,
each point in $P$ is either a vertex in $A$ or $B$
in one of the two parallel lines,
or the intersection of two or more segments in $\S$
between the two parallel lines.
Observe that during the execution of the greedy algorithm,
each intersection between (but in neither of) the two parallel lines
is incident to at most one segment from the subset of segments
incident to the vertices in $B_i$, $2 \le i \le k$;
similar to the vertices in $A$, these intersections are never selected by
the greedy algorithm.
Thus the greedy algorithm still selects the vertices in $B$
to cover the segments,
and its approximation ratio is still $\Omega(\log{n})$ by the same analysis.
\begin{figure}[htb]
\centering\includegraphics[scale=0.75]{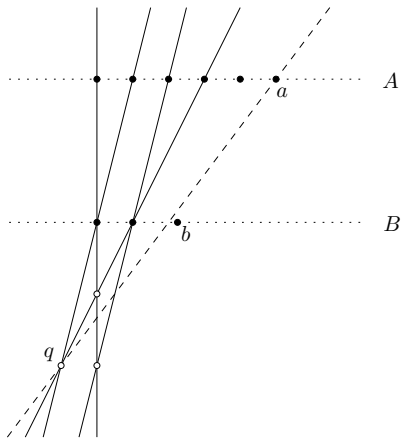}
\caption{Adapting the lower bound construction for {\sc Covering Segments by Points} 
to {\sc Covering Lines by Points}.}
\label{fig:greedy2}
\end{figure}

We next adapt this construction further to the problem
{\sc Covering Lines by Points}:
given a set $\L$ of $n$ lines in the plane,
find a set $P$ of points of the minimum cardinality such that
each line in $\L$ contains at least one point in $P$.
Since {\sc Covering Lines by Points} and 
our original problem {\sc Covering Points by Lines} 
are exact duals of each other, any lower bound we obtain for the
former is also a lower bound for the latter.

Refer to Figure~\ref{fig:greedy2}.
The straightforward part of the adaptation simply extends each segment
in $\S$ to a line in $\L$. This leads to more intersections, however,
above and below the two parallel lines.
As in the construction for {\sc Covering Segments by Points},
we place the vertices in $A$ evenly in the top line,
with unit distance between consecutive points.
For the vertices in $B$, however, we place them \emph{almost} evenly in
the bottom line, with \emph{near-unit} distance between consecutive
points (for convenience),
such that the following property is satisfied:
\begin{quote}
P1: Any intersection of the lines in $\L$,
if it is not a vertex in $A$ or $B$ in one of the two parallel lines,
it is incident to exactly two lines in $\L$.
\end{quote}

To ensure this property, we place the vertices in $B$ incrementally
as follows.
Let $B'$ be the subset of vertices in $B$ that have been placed,
and let $\L'$ be the subset of lines in $\L$ incident to $B'$.
Let $Q$ be the set of points that are intersections of the lines in $\L'$
but are not vertices in $A$ or $B$.
For each point $q$ in $Q$, and for each vertex $a$ in $A$,
mark the intersection of the bottom line and the line through $q$ and $a$.
Place the next vertex $b$ in $B$ in the bottom line to avoid such marks.

Due to the property P1, the greedy algorithm selects vertices in
$B_k,\ldots,B_3$ as before. Then, to cover the $2|B_2|$ lines incident to $B_2$,
it may select intersections not in the two parallel lines, but the number
of points it selects is at least $2|B_2|/2 = |B_2|$
since these lines are in general position.
Consequently, the same $\Omega(\log n)$ lower bound follows and this
completes the proof of Theorem~\ref{thm:greedy}.


\section{APX-hardness of {\sc Covering Points by Lines} and 
{\sc Maximum Point Coverage by Lines}}

In this section we prove
Theorems~\ref{thm:min} and~\ref{thm:max}.
Given a set $V$ of $n$ variables and a set $C$ of $m$ clauses,
where each variable has exactly $p$ literals (in $p$ different clauses)
and each clause is the disjunction of exactly $q$ literals
(of $q$ different variables),
{\sc E$p$-Occ-Max-E$q$-SAT} is the problem of finding an assignment of
the variables in $V$ that satisfies the maximum number of clauses in $C$.
Note that $pn = qm$.
Berman and Karpinski~\cite{BK99} showed that
even the simplest version of this problem,
{\sc E$3$-Occ-Max-E$2$-SAT}, is APX-hard,
and moreover this holds even if
the three literals of each variable
are neither all positive nor all negative;
see also~\cite{BK03} for the current best approximation lower bounds for
the many variants of {\sc E$p$-Occ-Max-E$q$-SAT} and related problems.
We prove that both {\sc Covering Points by Lines} and {\sc Maximum Point Coverage
by Lines} are APX-hard by two gap-preserving reductions from {\sc E$3$-Occ-Max-E$2$-SAT}
(Lemmas~\ref{lem:max} and~\ref{lem:min}, respectively). 

Let $(V,C)$ be an instance of {\sc E$3$-Occ-Max-E$2$-SAT},
where $V$ is a set of $n$ variables $v_i$, $1 \le i \le n$,
and $C$ is a set of $m$ clauses $c_j$, $1 \le j \le m$.
We construct a set $P$ of $4n + m$ points,
including four variable points $v_i^1,v_i^2,v_i^3,v_i^4$
for each variable $v_i$, $1 \le i \le n$,
and one clause point $c_j^*$ for each clause $c_j$, $1 \le j \le m$.
Assume that the three literals of each variable
are neither all positive nor all negative.
Then each variable has either two positive literals and one negative literal,
or two negative literals and one positive literal.
We place the point set $P$ in the plane 
(an example appears in Figure~\ref{fig:line})
such that no line goes through more than two points in $P$
except in the following two cases:
\begin{enumerate} \itemsep -1pt
\item
If a variable $v_i$ has two positive literals in $c_r$ and $c_s$,
respectively,
and has one negative literal in $c_t$,
then $c_r^*,v_i^1,v_i^2$ are collinear,
$c_s^*,v_i^3,v_i^4$ are collinear,
and $c_t^*,v_i^1,v_i^3$ are collinear.
\item
If a variable $v_i$ has two negative literals in $c_r$ and $c_s$,
respectively,
and has one positive literal in $c_t$,
then $c_r^*,v_i^1,v_i^3$ are collinear,
$c_s^*,v_i^2,v_i^4$ are collinear,
and $c_t^*,v_i^1,v_i^2$ are collinear.
\end{enumerate}

\begin{figure}[htb]
\centering\includegraphics[scale=0.85]{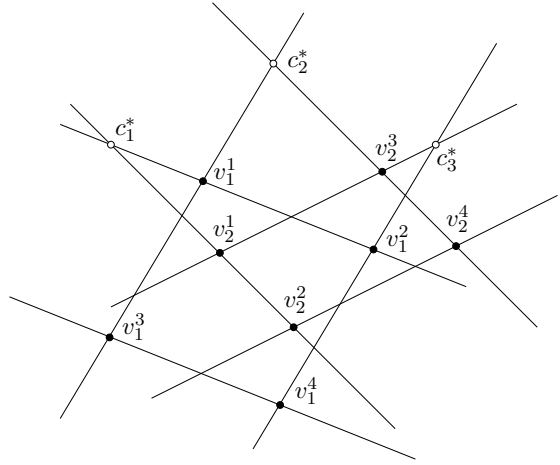}
\caption{%
The set $P$ of $4\times 2 + 3 = 11$ points
corresponding to the {\sc E$3$-Occ-Max-E$2$-SAT} instance $(V,C)$ where
$c_1 = v_1 \lor v_2$,
$c_2 = \bar v_1 \lor v_2$,
$c_3 = \bar v_1 \lor \bar v_2$.}
\label{fig:line}
\end{figure}

For any set $\L$ of lines,
let $\L_i \subseteq \L$ denote the subset of lines
incident to the four variable points $v_i^1,v_i^2,v_i^3,v_i^4$
of the variable $v_i$.
For each variable $v_i$,
and for each pair of indices $\{r,s\} \subset \{1,2,3,4\}$,
let $e_i^{r,s}$ denote the line through the two points $v_i^r$ and $v_i^s$.
We say that a set $\L$ of lines is \emph{canonical} if for each variable $v_i$,
$|\L_i| \le 2$ and
$\L_i \subseteq \{e_i^{1,2},e_i^{3,4},e_i^{1,3},e_i^{2,4}\}$,
and moreover,
if $|\L_i| = 2$, then $\L_i$ is
either $\{e_i^{1,2},e_i^{3,4}\}$
or $\{e_i^{1,3},e_i^{2,4}\}$.
The following lemma is used by both reductions:

\begin{lemma}\label{lem:canonical}
Any set $\L$ of $k$ lines that cover $x$ points in $P$
can be transformed into a canonical set of at most $k$ lines that
cover at least $x$ (possibly different) points in $P$. 
\end{lemma}
\begin{proof}
Consider an arbitrary function $f:P\to\L$ that maps each point $p \in P$
to a line $\ell \in \L$ incident to $p$;
if $p$ is not incident to any line in $\L$,
then $p$ is unmapped, \ie, $f(p)$ is undefined.
For each variable $v_i$, let $\L_i^f \subseteq \L$ denote the subset of
lines to which $f$ maps the four variable points $v_i^1,v_i^2,v_i^3,v_i^4$.
Clearly, $|\L_i^f| \le 4$.
Note that $\L_i^f \subseteq \L_i$ but in general $\L_i^f$ is not necessarily
the same as $\L_i$, because a variable point of $v_i$ may be incident to
multiple lines in $\L$ and is mapped by $f$ to at most one of them.
In the following, we will transform $\L$ and update $f$ accordingly,
until $\L_i^f = \L_i$ for all variables $v_i$
and $\L$ is canonical.
Initially, every covered point is mapped to some line.
During the transformation, we maintain the invariant that
every mapped point is covered by some line
(but not necessarily every covered point is mapped to some line)
and the number of mapped points is non-decreasing.

Categorize each line $\ell \in \L$ in one of three types according to $f$:
if there are two variable points of the same variable mapped to $\ell$,
then $\ell$ is of type $2$;
if there are no variable points mapped to $\ell$,
then $\ell$ is of type $0$;
otherwise, $\ell$ is of type $1$.
Note that each type-$1$ line either has only one variable point mapped to it,
or has two variable points of different variables mapped to it.

In the first step,
we transform $\L$ until $|\L_i^f| \le 2$ for each variable $v_i$.
If $|\L_i^f| > 2$ for some variable $v_i$,
then $\L_i^f$ must include at least two lines of type $1$.
Note that each type-$1$ line in $\L_i^f$
has exactly one variable point of $v_i$ and at most one other point
(either some clause point or a variable point of some other variable)
mapped to it.
As long as $\L_i^f$ includes two type-$1$ lines,
we replace them by a line of type $2$
(through the two variable points of $v_i$
previously mapped to the two type-$1$ lines)
and at most one other line of type $0$ or $1$
(through the at most two other points, if any,
previously mapped to the two type-$1$ lines),
and then update the function $f$ accordingly
(so that the points previously mapped to the two type-$1$ lines
are mapped to the lines that replace them).
This replacement reduces $|\L_i^f|$ by $1$
(because the other line, if any, does not have any variable points of $v_i$
mapped to it), and does not increase $|\L_j^f|$ for any~$j \neq i$.
Repeat such replacement whenever applicable.
Eventually we have $|\L_i^f| \le 2$ for every variable $v_i$.

In the second step, 
we transform $\L$ until no lines of type $0$ are incident to variable
points. Consider any line $\ell$ of type $0$.
If $\ell$ is incident to two clause points,
then by construction it is not incident to any variable point.
Otherwise $\ell$ is incident to at most one clause point,
and hence can be rotated, if necessary, to avoid all variable points.
Note that the function $f$ and the subsets $\L_i^f$
are not changed during this step.

In the third step,
we transform $\L$ until it contains no lines of type $1$.
Consider any line $\ell$ of type $1$ in $\L_i^f$,
with a variable point of $v_i$, say $v_i^r$, mapped to it.
Since $|\L_i^f| \le 2$,
there is at most one other line besides $\ell$ in $\L_i^f$,
with at most two other variable points of $v_i$ mapped to it.
It follows that at least one of the four variable points of $v_i$,
say $v_i^s$, is not mapped to any line in $\L$.
Replace $\ell$ by the line $e_i^{r,s}$, and update the function $f$ accordingly
(first unmap the at most two points previously mapped to $\ell$,
including $v_i^r$,
then map both $v_i^r$ and $v_i^s$ to the line $e_i^{r,s}$ of type $2$).

In the four step,
we transform $\L$ by considering two cases for each line of type $2$:
\begin{enumerate} \itemsep -1pt

\item
A line $\ell$ of type $2$ is in $\L_i^f$ with $|\L_i^f| = 1$.
Then $\ell$ is the only line with variable points of $v_i$ mapped to it,
and
$\ell \in \{e_i^{1,2},e_i^{3,4},e_i^{1,3},e_i^{2,4},e_i^{1,4},e_i^{2,3}\}$.
If $\ell \in \{e_i^{1,4},e_i^{2,3}\}$,
we replace $\ell$ by any line in
$\{e_i^{1,2},e_i^{3,4},e_i^{1,3},e_i^{2,4}\}$.

\item
Two lines $\ell_1$ and $\ell_2$ of type $2$ are in $\L_i^f$ with $|\L_i^f| = 2$.
Then $\{\ell_1, \ell_2\}$ must be
either $\{e_i^{1,2},e_i^{3,4}\}$,
or $\{e_i^{1,3},e_i^{2,4}\}$,
or $\{e_i^{1,4},e_i^{2,3}\}$.
If the two lines are $\{e_i^{1,4},e_i^{2,3}\}$,
we replace them by
either $\{e_i^{1,2},e_i^{3,4}\}$
or $\{e_i^{1,3},e_i^{2,4}\}$, arbitrarily.
\end{enumerate}
After each replacement,
we update $f$ accordingly.
This completes the transformation of $\L$ into canonical form.
\end{proof}

For the reduction to {\sc Maximum Point Coverage by Lines},
we have the following lemma about the construction:

\begin{lemma}\label{lem:max-iff}
There exists an assignment of the variables in $V$
that satisfies at least $w$ clauses in $C$
if and only if
there exists a set of $2n$ lines
that cover at least $4n+w$ points in $P$.
\end{lemma}

\begin{proof}
We first prove the direct implication.
Let $g:V \to \{\textrm{true}, \textrm{false}\}$
be an assignment that satisfies at least $w$ clauses in $C$.
For each variable $v_i \in V$, $1 \le i \le n$,
select two lines:
if $g(v_i) = \textrm{true}$,
select the line through $v_i^1$ and $v_i^2$
and the line through $v_i^3$ and $v_i^4$;
if $g(v_i) = \textrm{false}$,
select the line through $v_i^1$ and $v_i^3$
and the line through $v_i^2$ and $v_i^4$.
By construction,
these $2n$ lines cover not only all $4n$ variable points,
but also the at least $w$ clause points for the satisfied clauses.

We next prove the reverse implication.
Let $\L$ be a set of $2n$ lines that cover at least $4n+w$ points in $P$.
We will construct an assignment of the variables in $V$
that satisfies at least $w$ clauses in $C$.
By Lemma~\ref{lem:canonical}, we can assume that $\L$ is canonical.
Consider any line $\ell \in \L$.
If $\ell$ is not incident to any variable point,
then it is incident to at most two clause points,
and can be replaced by a line through two variable points of some variable
while keeping $\L$ in canonical form.
Repeat such replacement whenever applicable.
Eventually $\L$ includes exactly $2n$ lines
incident to all $4n$ variable points
and at least $w$ clause points.
Compose an assignment
$g:V \to \{\textrm{true}, \textrm{false}\}$
by setting $g(v_i)$ to true if $\L_i = \{e_i^{1,2},e_i^{3,4}\}$
and to false if $\L_i = \{e_i^{1,3},e_i^{2,4}\}$.
Then by construction $g$ satisfies at least $w$ clauses.
\end{proof}

For the reduction to {\sc Covering Points by Lines},
we have the following two lemmas analogous to Lemma~\ref{lem:max-iff}
about the construction:

\begin{lemma}\label{lem:min-direct}
If there exists an assignment of the variables in $V$
that satisfies at least $w$ clauses in $C$,
then there exists a set of at most $2n + \lceil (m - w) / 2 \rceil$ lines
that cover all points in $P$.
\end{lemma}

\begin{proof}
Let $g:V \to \{\textrm{true}, \textrm{false}\}$
be an assignment that satisfies at least $w$ clauses in $C$.
For each variable $v_i \in V$, $1 \le i \le n$,
select two lines:
if $g(v_i) = \textrm{true}$,
select the line through $v_i^1$ and $v_i^2$
and the line through $v_i^3$ and $v_i^4$;
if $g(v_i) = \textrm{false}$,
select the line through $v_i^1$ and $v_i^3$
and the line through $v_i^2$ and $v_i^4$.
By construction,
these $2n$ lines cover not only all $4n$ variable points,
but also the at least $w$ clause points for the satisfied clauses.
To cover the remaining at most $m - w$ clause points
for the unsatisfied clauses,
we pair them up arbitrarily and use at most $\lceil (m - w) / 2 \rceil$
additional lines.
\end{proof}

\begin{lemma}\label{lem:min-reverse}
If there exists a set of at most $2n + \lfloor (m - w) / 2 \rfloor$ lines
that cover all points in $P$,
then there exists an assignment of the variables in $V$
that satisfies at least $w$ clauses in $C$.
\end{lemma}

\begin{proof}
Let $\L$ be a set of at most
$2n + \lfloor (m - w) / 2 \rfloor$ lines that cover all points in $P$.
We will construct an assignment of the variables in $V$
that satisfies at least $w$ clauses in $C$.
By Lemma~\ref{lem:canonical}, we can assume that $\L$ is canonical.
Since all points are covered, this requires that $|\L_i| = 2$ for each
$i=1,\ldots,n$, thus $\L$ includes exactly $2n$ lines incident to the
variable points. 
These $2n$ lines must cover at least $w$ clause points
because the other at most $\lfloor (m-w)/2 \rfloor$ lines in $\L$ can
cover at most $m - w$ clause points.
Compose an assignment
$g:V \to \{\textrm{true}, \textrm{false}\}$
by setting $g(v_i)$ to true if $\L_i = \{e_i^{1,2},e_i^{3,4}\}$
and to false if $\L_i = \{e_i^{1,3},e_i^{2,4}\}$.
Then  by construction $g$ satisfies at least $w$ clauses.
\end{proof}

The following lemma implies that {\sc Maximum Point Coverage By Lines} is APX-hard:

\begin{lemma}\label{lem:max}
For any $\eps$, $0 < \eps < \frac15$,
if {\sc Maximum Point Coverage by Lines} admits a polynomial-time approximation algorithm
with ratio $1 - \eps$,
then {\sc E$3$-Occ-Max-E$2$-SAT} admits a polynomial-time approximation algorithm
with ratio $1 - 5\,\eps$.
\end{lemma}

\begin{proof}
Let $(V,C)$ be an instance of {\sc E$3$-Occ-Max-E$2$-SAT} with $n$ variables
and $m$ clauses, where $3n = 2m$.
Consider the following algorithm:
first construct a set $P$ of points from $(V,C)$ 
(refer back to Figure~\ref{fig:line}) and set
$k$ (the number of lines) to $2n$; 
then run the $(1-\eps)$-approximation algorithm for
{\sc Maximum Point Coverage by Lines}
on the instance $(P,k)$ to obtain a set $\L$ of lines,
and finally compose an assignment
$g:V \to \{\textrm{true}, \textrm{false}\}$
as in the reverse implication of Lemma~\ref{lem:max-iff}.
The algorithm can clearly be implemented in polynomial time.
It remains to analyze its approximation ratio.

Let $x^*$ be the maximum number of points in $P$
that can be covered by any set of $k = 2n$ lines.
Clearly, $x^* \le 4n + m$.
Let $w^*$ be the maximum number of clauses in $C$ that can be satisfied by
any assignment of $V$.
Observe that $w^* \ge \frac34 m$
since a random assignment of each variable independently to either true
or false with equal probability $\frac12$ satisfies each disjunctive clause
of two literals with probability $1 - (\frac12)^2 = \frac34$.
Recall that $3n = 2m$.
Thus we have
\begin{equation}\label{eq:x*w*}
x^* \le 4n + m = \frac{11}3 m \le \frac{11}3\,\frac43 w^* = \frac{44}9 w^*
< 5\,w^*.
\end{equation}

Let $x'$ be the number of points in $P$ covered by the lines in $\L$.
Let $w'$ be the number of clauses in $C$ satisfied by the assignment $g$.
Lemma~\ref{lem:max-iff} implies that
$w^* = x^* - 4n$
and by the reverse implication in Lemma~\ref{lem:max-iff} we have 
$w' \ge x' - 4n$.
Thus
$w^* - w' \le x^* - x'$.
It then follows from~\eqref{eq:x*w*} that
$$
\frac{w^* - w'}{w^*}
\le 5\,\frac{x^* - x'}{x^*}.
$$
The $(1-\eps)$-approximation algorithm for {\sc Maximum Point Coverage by Lines}
guarantees the relative error bound
$(x^* - x')/x^* \le \eps$.
So we have
$(w^* - w')/w^* \le 5\,\eps$
and hence $w'/w^* \ge 1 - 5\,\eps$, as desired.
\end{proof}

The following lemma, analogous to Lemma~\ref{lem:max},
implies that {\sc Covering Points by Lines} is APX-hard:

\begin{lemma}\label{lem:min}
For any $\eps$, $0 < \eps < \frac1{10}$,
if {\sc Covering Points by Lines} admits a polynomial-time approximation algorithm
with ratio $1 + \eps$,
then {\sc E$3$-Occ-Max-E$2$-SAT} admits a polynomial-time approximation algorithm
with ratio $1 - 10\,\eps$.
\end{lemma}

\begin{proof}
Let $(V,C)$ be an instance of {\sc E$3$-Occ-Max-E$2$-SAT} with $n$ variables
and $m$ clauses, where $3n = 2m$.
Let $w^*$ be the maximum number of clauses in $C$ that can be satisfied by
any assignment of $V$.
We have $w^* \ge \frac34 m$
since a random assignment of each variable independently to either true
or false with equal probability $\frac12$ satisfies each disjunctive clause
of two literals with probability $1 - (\frac12)^2 = \frac34$.
Without loss of generality, we assume that $w^* \ge 1/\eps$,
since otherwise the instance $(V,C)$ would have size $O(1/\eps)$
and would admit a straightforward brute-force algorithm
running in $2^{O(1/\eps)}$ time,
which is constant time for any fixed $\eps > 0$.

Under the assumption that $w^* \ge 1/\eps$,
we have the following algorithm for {\sc E$3$-Occ-Max-E$2$-SAT}:
first construct a set $P$ of points from $(V,C)$
(refer back to Figure~\ref{fig:line}),
then run the $(1+\eps)$-approximation algorithm for {\sc Covering Points by Lines} on $P$
to obtain a set $\L$ of lines,
and finally compose an assignment
$g:V \to \{\textrm{true}, \textrm{false}\}$
as in the proof of Lemma~\ref{lem:min-reverse}.
The algorithm can clearly be implemented in polynomial time.
It remains to analyze its approximation ratio.

Let $k^*$ be the minimum cardinality of any set of lines that cover $P$.
It is easy to see that $k^* \le 2n + m$.
Recall that $w^* \ge \frac34m$ and $3n = 2m$.
Thus we have
\begin{equation}\label{eq:k*w*}
k^* \le 2n + m \le 3m \le 4 w^*.
\end{equation}
Lemma~\ref{lem:min-direct} implies that
$k^* \le 2n + \lceil (m-w^*)/2 \rceil$.
It follows that
$k^* \le 2n + (m-w^*+1)/2$
and hence
\begin{equation}\label{eq:w*k*}
w^* \le 2(2n - k^*) + m + 1.
\end{equation}

Let $k'$ be the number of lines in $\L$.
Let $w'$ be the number of clauses in $C$ that are satisfied by the
assignment $g$.
Put $w := m - 1 - 2(k' - 2n)$.
Then
\begin{equation}\label{eq:k'w}
k' =  2n + (m-1-w)/2 = 2n + \lfloor (m-w)/2 \rfloor.
\end{equation}
Note that
$w := m - 1 - 2(k' - 2n)$ is the smallest integer (there are two such integers)
satisfying the equation
$k' =  2n + \lfloor (m-w)/2 \rfloor$.
If $w' < w$, then by (the contrapositive of) Lemma~\ref{lem:min-reverse}
we would have
$k' > 2n + \lfloor (m-w)/2 \rfloor$,
which contradicts~\eqref{eq:k'w}.
So we must have $w' \ge w$, that is,
\begin{equation}\label{eq:w'k'}
w' \ge m - 1 - 2(k' - 2n).
\end{equation}
From \eqref{eq:w*k*} and \eqref{eq:w'k'}, we have
$$
w^* - w'
\le 2(2n - k^*) + m + 1 - m + 1 + 2(k' - 2n)
= 2(k' - k^*) + 2,
$$
and hence
$$
\frac{w^* - w'}{w^*}
\le \frac{2(k' - k^*)}{w^*} + \frac2{w^*}
\le \frac{8(k' - k^*)}{k^*} + \frac2{w^*},
$$
where the second inequality follows from~\eqref{eq:k*w*}.
The $(1+\eps)$-approximation algorithm for {\sc Covering Points by Lines}
guarantees the relative error bound
$(k' - k^*)/k^* \le \eps$.
Recall our assumption that
$w^* \ge 1/\eps$.
Consequently we have
$(w^* - w')/w^* \le 8\,\eps + 2\,\eps = 10\,\eps$
and hence $w'/w^* \ge 1 - 10\,\eps$, as desired.
\end{proof}

\paragraph{Remark.} For simplicity, we did not attempt obtaining
the best multiplicative constant factors of $\eps$ in the previous two
lemmas. Those expressions can be improved.

\section{APX-hardness of {\sc Minimum-Link Covering Tour}}

In this section we prove Theorem~\ref{thm:covering}.
We show that {\sc Minimum-Link Covering Tour} is APX-hard
by a gap-preserving reduction from {\sc Covering Points by Lines}\footnote{%
Arkin~\etal~\cite{AMP03} proved the NP-hardness of
{\sc Minimum-Link Covering Tour} by a reduction from the same problem
{\sc Covering Points by Lines}, 
but since their reduction is not gap-preserving,
their proof does not immediately imply the APX-hardness of
{\sc Minimum-Link Covering Tour} even if {\sc Covering Points by
  Lines} was known to be APX-hard. 
It is quite likely, however, that their construction can be combined with
our construction in the proof of Theorem~\ref{thm:min}
to obtain a gap-preserving reduction
directly from {\sc E$3$-Occ-Max-E$2$-SAT} to {\sc Minimum-Link Covering Tour}.},
which was proved to be APX-hard in Theorem~\ref{thm:min}.

Let $P$ be a set of $n$ points for the problem {\sc Covering Points by Lines}.
We will construct a set $Q$ of $3n$ points for the problem
{\sc Minimum-Link Covering Tour},
such that $P$ can be covered by $k$ lines if and only if
$Q$ admits a covering tour with $3k$ segments.

\begin{figure}[htb]
\centering\includegraphics[scale=0.9]{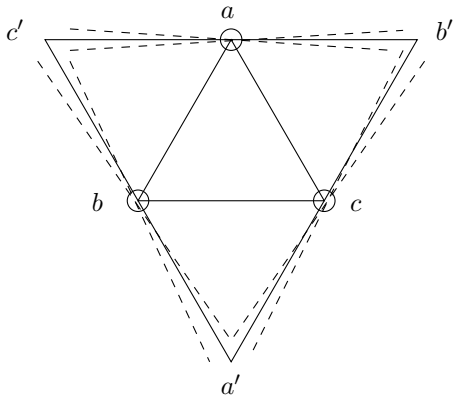}
\caption{%
The construction of $Q$ from $P$.}
\label{fig:triangle}
\end{figure}

Refer to Figure~\ref{fig:triangle}.
By an affine transformation, we first transform $P$
into a set $P'$ of $n$ points such that
(i) $P'$ is enclosed in a circle of some small radius $r$, say, $r = 1/100$;
(ii) the angle between any two lines $\ell_1$ and $\ell_2$,
each incident to at least two points in $P'$,
is at most some small angle $\theta$, say, $1^\circ$.
Now take an equilateral triangle $abc$ of side length $1$ inscribed
in an equilateral triangle $a'b'c'$ of side length $2$,
where the three vertices of the smaller triangle are the midpoints of
the three edges of the larger triangle.
The point set $Q$ is the union of three rotated copies of $P'$ that we refer
to as the three clusters,
one cluster near each vertex of $abc$,
such that the circle of radius $r$ enclosing each cluster is centered
at the vertex,
and all lines through at least two points in the cluster are at angles at most
$\theta$ from the edge of $a'b'c'$ that contains the vertex.

\begin{lemma}
There exists a set of $k$ lines that cover all points in $P$
if and only if
there exists a covering tour with $3k$ segments for $Q$.
\end{lemma}
\begin{proof}
We first prove the direct implication.
Let $\L$ be a set of $k$ lines that cover all points in $P$.
Then by the affine transformation,
we have a set $\L'$ of $k$ lines that cover all points in $P'$,
and the $3k$ lines in the three copies of $\L'$
corresponding to the three copies of $P'$ cover all points in $Q$.
These $3k$ lines can obviously be linked into a covering tour
with $3k$ segments,
where any three consecutive segments are from three different clusters,
and the turns between consecutive segments are near the vertices of $a'b'c'$.

We next prove the reverse implication.
Let $\T$ be a covering tour with $3k$ segments for $Q$,
and let $\L_\T$ be the set of at most $3k$ lines
supporting the $3k$ segments in $\T$.
A line is an \emph{intra-cluster} line if the points in $Q$ that are covered
by it, if any, are all from the same cluster;
it is an \emph{inter-cluster} line otherwise.
By construction,
each inter-cluster line covers exactly two points, from two different clusters.
Let $l_1$ (respectively, $l_2$) be the number of
intra-cluster (respectively, inter-cluster) lines in $\L_\T$;
then $l_1 + l_2 \le 3k$.
Let $n_a,n_b,n_c$ be the numbers of points in clusters near $a,b,c$,
respectively, that are covered by the inter-cluster lines;
then $n_a + n_b + n_c = 2 l_2$.
Since any two points in the same cluster can be covered by some intra-cluster
line, these $n_a + n_b + n_c$ points can be covered by at most
$\lceil n_a/2 \rceil + \lceil n_b/2 \rceil + \lceil n_c/2 \rceil$
intra-cluster lines (instead of $l_2$ inter-cluster lines).
Since the sum of the three numbers $n_a,n_b,n_c$ is even,
we have either two of them odd and one even, or all three of them even.
Thus we have
$\lceil n_a/2 \rceil + \lceil n_b/2 \rceil + \lceil n_c/2 \rceil
\le (n_a + n_b + n_c) /2 + 1 = l_2 + 1$.
It follows that $Q$ can be covered by at most $l_1 + (l_2 + 1) \le 3k + 1$
intra-cluster lines, and hence at least one of the three copies of $P'$
can be covered by at most $\lfloor (3k+1)/3 \rfloor = k$ lines.
The corresponding $k$ lines obtained by reversing the affine transformation
cover all points in $P$.
\end{proof}

From the above lemma, we can easily prove the following lemma similar to
Lemmas~\ref{lem:max} and~\ref{lem:min},
which implies the APX-hardness of {\sc Minimum-Link Covering Tour}:

\begin{lemma}\label{lem:covering}
For any $\rho \ge 1$,
if {\sc Minimum-Link Covering Tour} admits a polynomial-time approximation algorithm
with ratio $\rho$,
then {\sc Covering Points by Lines} admits a polynomial-time approximation algorithm
with ratio $\rho$.
\end{lemma}

\section{NP-hardness of {\sc Minimum-Link Spanning Tour}}

In this section we prove Theorem~\ref{thm:spanning}.
We show that {\sc Minimum-Link Spanning Tour} is NP-hard by a reduction from
a variant of the NP-hard problem
{\sc Hamiltonian Circuit in Cubic Graphs}~\cite{GJS74},
in which the input consists of not only a cubic graph $G$
but also some edge $e$ of $G$ that is required to be part of the Hamiltonian
circuit.
A simple Turing reduction shows that this variant is still NP-hard.

Let $G$ be a cubic graph with $n$ vertices and $m$ edges, where $3n = 2m$.
Let $e = \{s,t\}$ be the edge of $G$ that is required to be part of the circuit.
We first obtain a graph $G'$ from $G$ by removing the edge $e$
then adding two dummy vertices $s'$ and $t'$ with two new edges
$\{s, s'\}$ and $\{t, t'\}$.
Then there exists a Hamiltonian circuit in $G$ containing the edge $e$
if and only if
there exists a Hamiltonian path in $G'$ from $s'$ to $t'$.
Observe that $G'$ has exactly $n+2$ vertices and exactly $m+1$ edges,
where $n+2 < m+1$,
and moreover every vertex except $s'$ and $t'$ has degree $3$.

We next construct a set $P$ of $n+2+m+1$ points,
one vertex point for each vertex,
and one edge point for each edge in $G'$.
The $n+2$ vertex points are in some arbitrary convex position,
say, on a circle.
The $m+1$ edge points are in the interior of the convex hull
of the $n+2$ vertex points.
Moreover, for each edge $e = \{ u, v \}$,
the edge point of $e$ is in the line through the two vertex points
of $u$ and $v$ (\ie, in the interior of the segment $uv$),
and is not in any other line containing more than two points in $P$.

The reduction can clearly be implemented in polynomial time.
Then the following lemma establishes the NP-hardness
of {\sc Minimum-Link Spanning Tour}:

\begin{lemma}
There exists a Hamiltonian path in $G'$ from $s'$ to $t'$
if and only if
there exists a spanning tour with $m+2$ segments for $P$.
\end{lemma}

\begin{proof}
We first prove the direct implication.
Let $H$ be a Hamiltonian path in $G'$ from $s'$ to $t'$.
We will construct a spanning tour with $m+2$ segments for $P$.
Corresponding to the Hamiltonian path $H$
that visits all $n+2$ vertices in $G'$ using $n+1$ edges,
there is a polygonal chain of $n+1$ segments
that connect the $n+2$ vertex points in $P$ in the same order,
which also cover $n+1$ edge points.
The chain can be extended to visit the remaining $(m+1)-(n+1)=m-n$ edge points
in any order with $m-n$ segments,
and finally closed into a tour with another segment.
The total number of segments is $(n+1) + (m-n) + 1 = m+2$.

We next prove the reverse implication.
Let $\T$ be a spanning tour with $m+2$ segments for $P$.
We will find a Hamiltonian path in $G'$ from $s'$ to $t'$.
We first transform $\T$,
without increasing the number of segments,
into a canonical spanning tour that visits each vertex point exactly once.

\begin{figure}[htb]
\centering\includegraphics[scale=0.75]{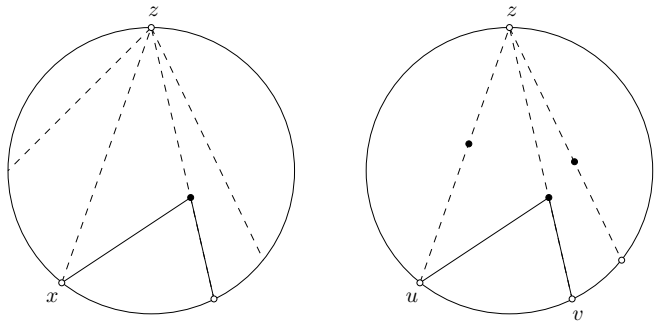}
\caption{%
Transforming $\T$ into a canonical spanning tour that visits each
vertex point exactly once.}
\label{fig:shortcut}
\end{figure}

Suppose that a vertex point $z$ is visited twice in $\T$.
Then there are two pairs of consecutive segments with turns at $z$.
Refer to Figure~\ref{fig:shortcut}.
If, out of these four segments incident to $z$,
there is a segment connecting some point $x$ directly to $z$
with no other point of $P$ in the interior of the segment,
then we can take a shortcut (as in Figure~\ref{fig:shortcut} left)
in the pair of segments including $xz$ to skip a visit to $z$.
Otherwise, each of the four segments must connect $z$ to the vertex
point of a neighbor of the vertex of $z$ in $G'$,
going through the corresponding edge point.
Recall that every vertex in $G'$ has degree at most $3$.
By the pigeonhole principle,
at least two of these four segments must be the same segment, say, $uz$.
If the two copies of $uz$ are consecutive
and form a turn at $z$ in $\T$,
then we can shorten both of them to skip a visit to $z$.
Otherwise,
one copy of $uz$ must form a turn at $z$
with some other segment, say $vz$,
and again we can take a shortcut (as in Figure~\ref{fig:shortcut} right)
to skip a visit of $z$.
Observe that in both cases,
$\T$ remains a spanning tour after the transformation.

Now observe that every edge point is
either in the interior of some segment between two vertex points,
or at a turning point between two consecutive segments.
Associate a cost of $1$ with each edge point.
For each edge point, charge its cost to the segments that contain it:
if it is in the interior of one segment, charge $1$ to the segment;
if it is at a turning point between two consecutive segments,
charge $1/2$ to each segment.
Observe that
\begin{itemize} \itemsep -1pt
\item
each segment between two vertex points is charged $1$ if the segment
contains an edge point, and is charged $0$ otherwise;
\item
each segment between two edge points is charged $1/2 + 1/2 = 1$;
\item
each segment between a vertex point and an edge point
(there must be at least two such segments along the tour
since there are more edge points than vertex points)
is charged exactly $1/2$.
\end{itemize}
Since the number of edge points in $P$ is $m+1$
and the number of segments in $\T$ is $m+2$,
we must have exactly $2$ segments charged $1/2$ each,
and exactly $m$ segments charged $1$ each,
so that $m+1 = 2\cdot 1/2 + m\cdot 1$.
It follows that (i)
there are exactly two segments between vertex points and edge points,
and (ii)
there is no segment connecting two vertex points and containing no edge point.
Condition (i) implies that the segments between vertex points
are consecutive in the tour.
Condition (ii) implies that these consecutive segments correspond to
a Hamiltonian path in $G'$.
Finally, this Hamiltonian path must have $s'$ and $t'$ as the two ends
because each of them has exactly one neighbor.
\end{proof}

\section{APX-hardness of {\sc Min-Max-Turn Hamiltonian Tour}
	and NP-hardness of {\sc Bounded-Turn-Minimum-Length Hamiltonian Tour}}

In this section we prove
Theorems~\ref{thm:hamiltonian} and~\ref{thm:hamiltonian2}.
We first show that {\sc Min-Max-Turn Hamiltonian Tour} is APX-hard
by a gap-preserving reduction from {\sc Covering Points by Lines},
which was proved to be APX-hard in Theorem~\ref{thm:min}.
Let $P$ be a set of $n$ points for the problem {\sc Covering Points by Lines}.
We will construct a set $Q$ of points for the problem
{\sc Min-Max-Turn Hamiltonian Tour}.
\begin{figure}[htb]
\centering\includegraphics[scale=0.6]{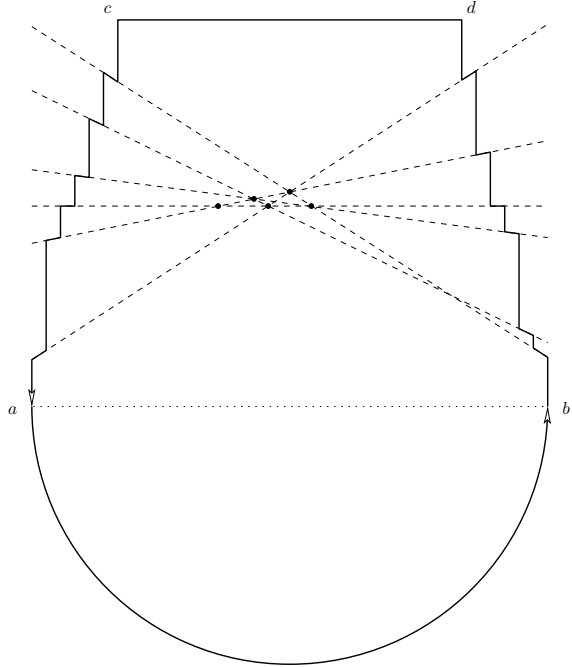}
\caption{%
The construction of the point set $Q$ for {\sc Min-Max-Turn Hamiltonian
Tour} ($s=4 \times 6 +3=27$ in this example).}
\label{fig:hamiltonian}
\end{figure}

Let $\L_P$ be the set of at most $n \choose 2$ lines determined by $P$,
where each line in $\L_P$ goes through at least two points in $P$
(we assume without loss of generality that no line in $\L_P$ is vertical).
Let $\alpha_P$ be the minimum turning angle determined
by any three non-collinear points in $P$;
we will show later in Lemma~\ref{lem:lattice}
that $1/\alpha_P$ may be assumed to be polynomial in $n$.
Let $\alpha := \alpha_P/n^\mu$ for some suitable constant $\mu \ge 1$.

Refer to Figure~\ref{fig:hamiltonian}.
We first construct a simple closed curve $\gamma$ that is composed by
a half-circle and a polygonal chain
joined at their endpoints $a$ and $b$.
The polygonal chain consists of
$s := 2|\L_P| + 2 (|\L_P| + 1) + 1 = O(n^2)$
segments, including $2$ segments from each line in $\L_P$,
$2(|\L_P| + 1)$ vertical segments, and $1$ horizontal segment $cd$.
Observe that $\gamma$ is \emph{monotone} in the horizontal direction,
in the sense that the intersection of every vertical line
with the region enclosed by $\gamma$ is a single line segment.

\begin{figure}[htb]
\centering\includegraphics[scale=0.8]{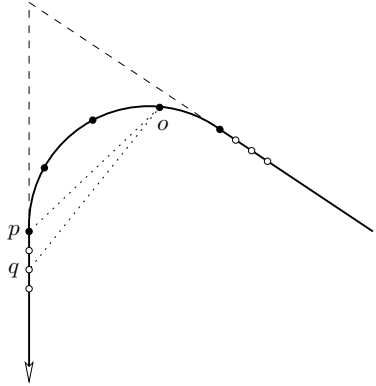}
\caption{%
Smoothing the corner between two consecutive segments
into a circular arc tangent to both segments
($t = 4$ and $n = 3$ in this example).}
\label{fig:corner}
\end{figure}

Refer to Figure~\ref{fig:corner}.
We next transform $\gamma$ into a simple closed curve $\tilde\gamma$
that is not only monotone but also \emph{smooth},
by smoothing the corner between every pair of consecutive segments
in the chain into a small circular arc tangent to both segments.
Then $\tilde\gamma$ is an alternating cycle of $s$ segments
(including the horizontal segment $cd$)
and $s$ circular arcs
(including the half-circle with diameter $ab$).

Put $t := \lceil 9n^2 \pi/\alpha \rceil$.
The point set $Q$ consists of all $n$ points in $P$
and $m := s(t+1+2n)$ points from the curve $\tilde\gamma$.
Take $t+1$ points (including the two endpoints)
from each of the $s$ circular arcs,
which divide any such arc into $t$ sub-arcs with the same central angle
at most $\theta := \pi/t \le \alpha / (9n^2)$.
Take $n$ more points from each of the two segments adjacent
to the circular arc, near the shared endpoints, such that
the following \emph{$\theta$-property} is satisfied:
the angle $\angle poq$,
where $p$ is an endpoint of the arc,
$o$ is any of the other $t$ points in the arc,
and $q$ is any of the $n$ points in the segment containing $p$,
is at most $\theta$.
Observe that the total number of points is $s(t+1+2n)$. 

We have the following two lemmas about the construction:

\begin{lemma}\label{lem:hamiltonian-direct}
If there exists a set of $k$ lines that cover all points in $P$,
for some $k \leq n$,
then there exists a Hamiltonian tour with maximum turning angle at most
$(k+2)\theta$ for $Q$.
\end{lemma}

\begin{proof}
Let $\L = \{\ell_1,\ldots,\ell_k\}$
be a set of $k$ lines that cover all points in $P$;
without loss of generality, $\L \subseteq \L_P$.
We will construct a Hamiltonian tour with maximum turning angle at most
$(k+2)\theta$ for $Q$.
Index the $m$ points in $Q$ taken from the curve $\tilde\gamma$
by the circular order of their locations along the curve:
$q_1,\ldots,q_m$.
Assign each point a color in $[0,k]$:
$q_i$ has color $i \bmod (k+1)$, $1 \le i \le m$.

The tour consists of $k+1$ rounds.
In the $i$th round, $1 \le i \le k$,
the tour follows the half-circle from $a$ to $b$ and continues along the
chain from $b$ to $d$ until it reaches the line $\ell_i$,
then takes a shortcut along $\ell_i$ from right to left
and continues along the chain from $c$ to $a$
until it reaches the half-circle again;
the tour visits each point of color $i$
in or below the line $\ell_i$ while following the curve,
and visits each point in $P$ that is covered by the line $\ell_i$
(if the point was not visited in previous rounds)
while taking the shortcut.
In the last round, the tour follows the curve entirely
to visit each point of color $0$,
and points of other colors not visited in previous rounds due to the shortcuts.

Consider any three consecutive points $p_1,p_2,p_3$ in the tour.
If the three points are all in some line $\ell_i$ supporting a shortcut,
then obviously $\mathrm{turn}(p_1,p_2,p_3) = 0$.
Otherwise, the three points must come from some sub-curve of $\tilde\gamma$
consisting of a circular arc and an adjacent segment.
The number of sub-arcs of this arc that are between $p_1$ and $p_3$
is at most $2(k+1)$;
each such sub-arc contributes half of its central angle
to the turning angle at $p_2$.
Taking into account the possibility that
the three points are not all in the arc
and using the $\theta$-property,
we have
$\mathrm{turn}(p_1,p_2,p_3) \le 2(k+1)\theta/2 + \theta = (k+2)\theta$.
\end{proof}

\begin{lemma}\label{lem:hamiltonian-reverse}
If there exists a Hamiltonian tour with maximum turning angle at most
$k\theta$ for $Q$, for some $k \leq n$,
then there exists a set of $k$ lines that cover all points in $P$.
\end{lemma}
\begin{proof}
Let $\tau$ be a Hamiltonian tour with maximum turning angle at most
$k\theta$ for $Q$.
We will find a set of $k$ lines that cover all points in $P$.

Break the tour $\tau$ into rounds,
such that each round consists of
some points in the half-circle followed by
some points not in the half-circle (\ie, in the chain or in $P$).
When the angle $\alpha$ is sufficiently small
(hence $k\theta \le k\alpha/(9n^2)$ is even smaller),
in each round the tour
must visit some points in the half-circle in order of their $x$-coordinates,
say, from left to right,
then visit points near some corners of the chain, from right to left.
While in the chain,
the tour may take shortcuts between non-consecutive corners,
but since the curve $\tilde\gamma$ is monotone,
it can take at most one \emph{crossing} shortcut from a corner on the $bd$ side
to a corner on the $ac$ side.
Only when taking such a crossing shortcut can the tour visit some points in $P$.
Moreover,
since $k\theta \le k\alpha/(9n^2) < \alpha_P$,
the points in $P$ that are visited during each crossing shortcut
must be collinear.

We next show that $\tau$ has at most $k$ rounds.
Index the $t+1$ points in the half-circle from left to right
by numbers from $0$ to $t$,
where $a$ has index $0$ and $b$ has index $t$.
Consider an arbitrary round.
Let $r$ be the number of points in the half-circle
that are visited in this round.
Let $i_1,\ldots,i_r$ be the indices of these points, from left to right.
Then we must have $i_1 \le 2k-1$ because otherwise
the turning angle at the point with index $i_1$ would be greater than
$(i_1 + 1)\theta/2 > 2k\,\theta/2 = k\theta$.
Similarly, we must have $i_r \ge t - (2k-1)$,
and $i_{j+2} - i_j \le 2k$ for each $j$, $1 \le j \le r-2$.
Counting in pairs, we have $r \ge 2\, r_2$,
where $r_2 = \bigl\lfloor \frac{t - (2k-1) - (2k-1)}{2k} \bigr\rfloor$.
It follows that the number of rounds in $\tau$ is at most
$$
\frac{t}{r}
\le \frac{t}{2 \bigl\lfloor \frac{t - (2k-1) - (2k-1)}{2k} \bigr\rfloor}
\le \frac{t}{2 \bigl( \frac{t - (2k-1) - (2k-1)}{2k} - 1 \bigr)}
= \frac{t}{t - 6k + 2}\,k.
$$
It is easy to check that $t \ge 9n^2 > 8k^2$
and hence $\frac{t}{t - 6k + 2} < \frac{k+1}k$.
Thus $\tau$ has at most $k$ rounds.
Finally,
since each round has at most one crossing shortcut
that can cover some collinear points in $P$,
all points in $P$ can be covered by $k$ lines.
\end{proof}

Unlike the reductions in previous sections,
our reduction to {\sc Min-Max-Turn Hamiltonian Tour} is numerically sensitive
because the construction depends on a small angle $\alpha_P$
and has points placed precisely on circular arcs.
Even if we reduce from a restricted version of
{\sc Covering Points by Lines},
where the coordinates are polynomial in the number of points
(it can be checked that our proof for the APX-hardness of
{\sc Covering Points by Lines} fulfills this restriction),
it is still not immediately clear that the reduction is polynomial.
To clarify this,
we prove a property concerning lattice points in the next lemma,
which implies that $1/\alpha_P$ may be assumed to be polynomial in
the lattice size $N$:

\begin{lemma}\label{lem:lattice}
Let $a$, $b$, and $c$ be three non-collinear points in the $[0,N] \times [0,N]$
section of the integer lattice, where $N \geq 10$. Then the turning
angle of the path $(a,b,c)$ at point $b$ is at least $1/(3 N^2)$.
\end{lemma}
\begin{proof}
Let $\beta$ denote the turning angle. We can assume that $0 < \beta <
\pi/2$, since otherwise the inequality holds.
Since $a,b,c$ are non-collinear lattice points,
the tangent function of the turning angle can be expressed as:
$$
\tan \beta = \tan(\beta_2 - \beta_1) =
\frac{\tan \beta_2 - \tan \beta_1}{1+ \tan \beta_1 \tan \beta_2},
$$
where
$\tan \beta_1 = \frac{m_1}{n_1} < \frac{m_2}{n_2}= \tan \beta_2$,
and $m_1, m_2, n_1, n_2$ are nonnegative integers less or equal to $N$.
Write $t_1= \tan \beta_1$ and $t_2= \tan \beta_2$.
We distinguish two cases depending on whether the product $t_1 t_2$ is
smaller or larger than $1$:

{\em Case 1}: $t_1 t_2 \leq 1$. We have
$$
\tan \beta = \frac{t_2 -t_1}{1 + t_1 t_2} \geq  \frac{t_2 -t_1}{2}
= \frac{1}{2} \left( \frac{m_2}{n_2} - \frac{m_1}{n_1} \right) =
\frac{m_2 n_1 - m_1 n _2}{2n_1 n_2} \geq \frac{1}{2n_1 n_2} \geq
\frac{1}{2 N^2}.
$$

{\em Case 2}: $t_1 t_2 \geq 1$. We have
$$
\tan \beta = \frac{t_2 -t_1}{1 + t_1 t_2} \geq  \frac{t_2 -t_1}{2t_1 t_2}
= \frac{1}{2} \left( \frac{1}{t_1} - \frac{1}{t_2} \right)
= \frac{1}{2} \left( \frac{n_1}{m_1} - \frac{n_2}{m_2} \right) =
\frac{m_2 n_1 - m_1 n _2}{2m_1 m_2} \geq \frac{1}{2m_1 m_2} \geq
\frac{1}{2 N^2}.
$$

Since $N \geq 10$ was assumed,
in both cases it follows that $\beta \geq 1/(3 N^2)$, as required.
\end{proof}

Having polynomial representation of the small angles and rational points
on circular arcs is a non-trivial problem~\cite{CDR92,Bu98}.
Without delving too much into technical details such as Lemma~\ref{lem:lattice},
we claim that for any $\delta$, $0 < \delta < 1$,
the construction can use integers polynomial in $n$
for the coordinates of all points in $Q$,
such that the angle determined by any three points in $Q$ deviates by
a multiplicative factor at least $1-\delta$ and at most $1+\delta$.
Consequently,
the reduction is strongly polynomial,
and we have the following approximate versions of
Lemmas~\ref{lem:hamiltonian-direct} and~\ref{lem:hamiltonian-reverse}:
\begin{description}

\item[Lemma~\ref{lem:hamiltonian-direct} (Approximate Version).]
If there exists a set of $k$ lines that cover all points in $P$,
for some $k \leq n$,
then there exists a Hamiltonian tour with maximum turning angle at most
$(1+\delta)(k+2)\theta$ for $Q$.

\item[Lemma~\ref{lem:hamiltonian-reverse} (Approximate Version).]
If there exists a Hamiltonian tour with maximum turning angle at most
$(1-\delta)k\theta$ for $Q$, for some $k \leq n$,
then there exists a set of $k$ lines that cover all points in $P$.

\end{description}

The following lemma shows that {\sc Min-Max-Turn Hamiltonian Tour}
is almost as hard to approximate as {\sc Covering Points by Lines}:

\begin{lemma}\label{lem:hamiltonian}
For any $\rho \ge 1$,
if {\sc Min-Max-Turn Hamiltonian Tour}
admits a polynomial-time approximation algorithm with ratio $\rho$,
then {\sc Covering Points by Lines} admits a polynomial-time
approximation algorithm with ratio $(1 + \eps)\rho$ for any $\eps > 0$.
\end{lemma}

\begin{proof}
Let $P$ be a set of $n$ points for {\sc Covering Points by Lines}.
Let $k^*$ be the minimum number of lines
necessary for covering all points in $P$.
Without loss of generality, we assume that $k^* \ge 6/\eps$, since otherwise
a brute-force algorithm can find $k^*$ lines to cover $P$
in $n^{O(1/\eps)}$ time, which is polynomial time for any fixed $\eps > 0$.

Let $\delta = \Theta(\eps)$ such that
$\frac{1+\delta}{1-\delta} = \frac{1+\eps}{1+\eps/2}$.
Under the assumption that $k^* \ge 6/\eps$, we have the following algorithm:
first construct a point-set $Q$ from $P$ (refer to Figure~\ref{fig:hamiltonian}),
and then run the $\rho$-approximation algorithm for {\sc Min-Max-Turn Hamiltonian Tour}
on $Q$ to obtain a tour of maximum turning angle at most
$(1-\delta)\kappa\theta$ for some $\kappa>0$,
and finally obtain a set $\L$ of at most $k = \lceil \kappa \rceil$ lines 
that cover $P$ by Lemma~\ref{lem:hamiltonian-reverse} (approximate version).

By Lemma~\ref{lem:hamiltonian-direct} (approximate version),
the minimum value of the maximum turning angle of a Hamiltonian tour
for $Q$ is at most $(1+\delta)(k^* + 2)\theta$
(\ie, there exists a Hamiltonian tour with maximum turning angle
bounded as such). 
The $\rho$-approximation algorithm for {\sc Min-Max-Turn Hamiltonian Tour}
guarantees that
$(1-\delta)\kappa\theta \le \rho(1+\delta)(k^* + 2)\theta$
and hence
$\kappa \le \rho\frac{1+\delta}{1-\delta}(k^* + 2)$.
By the assumption that $k^* \ge 6/\eps$,
we have 
$k = \lceil \kappa \rceil \le \rho\frac{1+\delta}{1-\delta}(k^* + 2) +
1 \le \rho\frac{1+\delta}{1-\delta}(k^* + 3) 
= \rho\frac{1+\delta}{1-\delta}(1 + 3/k^*) k^* \le
\rho\frac{1+\delta}{1-\delta}(1+\eps/2) k^* 
= \rho (1+\eps) k^*$.
\end{proof}

Since {\sc Covering Points by Lines} is APX-hard (Theorem~\ref{thm:min}),
the above lemma implies that {\sc Min-Max-Turn Hamiltonian Tour} is APX-hard too.
This completes the proof of Theorem~\ref{thm:hamiltonian}.

\medskip
Consider the decision versions of the two Hamiltonian Tour problems:
\begin{itemize} \itemsep -1pt
\item[(I)]
{\sc Min-Max-Turn Hamiltonian Tour} (Decision Problem):
Given $n$ points in the plane and an angle $\alpha \in [0,\pi]$,
decide whether there exists a Hamiltonian tour with maximum turning angle
at most $\alpha$.
\item[(II)]
{\sc Bounded-Turn-Minimum-Length Hamiltonian Tour} (Decision Problem):
Given $n$ points in the plane, an angle $\alpha \in [0,\pi]$,
and a positive number $L$,
decide whether there exists a Hamiltonian tour with maximum turning angle
at most $\alpha$ and length at most $L$.
\end{itemize}
Observe that the decision problem of 
{\sc Min-Max-Turn Hamiltonian Tour}
is a special case of the decision problem of
{\sc Bounded-Turn-Minimum-Length Hamiltonian Tour},
with the parameter $L$ set to some sufficiently large number,
say, $n$ times the diameter of the point set.
Thus the APX-hardness (indeed NP-hardness suffices)
of {\sc Min-Max-Turn Hamiltonian Tour}
implies that {\sc Bounded-Turn-Minimum-Length Hamiltonian Tour} is NP-hard.
This completes the proof of Theorem~\ref{thm:hamiltonian2}.

\medskip
It is interesting to note that,
while the decision problem of {\sc Bounded-Turn-Minimum-Length Hamiltonian Tour}
has both an \emph{angle} constraint and a \emph{length} constraint,
our proof of its NP-hardness above
(via the reduction from the decision problem of {\sc Min-Max-Turn Hamiltonian Tour})
effectively only uses the angle constraint.
That is, the problem is already hard with the angle constraint alone.
On the other hand,
if the turning angle is unrestricted, \ie, if $\alpha = \pi$,
then the problem {\sc Bounded-Turn-Minimum-Length Hamiltonian Tour}
is the same as the {\sc Euclidean Traveling Salesman Problem},
which is well known to be NP-hard with the length constraint alone.
Our proof
sheds light on a different aspect of the difficulty of the problem.

\section{Concluding remarks}

The obvious question left open by our work is whether 
{\sc Covering Points by Lines} admits an approximation algorithm with
constant ratio. 
Two other problems are finding approximation algorithms for 
{\sc Min-Max-Turn Hamiltonian Tour} and respectively, 
{\sc Bounded-Turn-Minimum-Length Hamiltonian Tour}.

\paragraph{Acknowledgment.}
The authors are grateful to an anonymous reviewer for suggesting the extension
to lines---in our Theorem~\ref{thm:greedy}---of the result
from~\cite{BLWM12} (Theorem~1, p.~1041) for line segments,
and to another anonymous reviewer for pertinent advice on polynomial
representation of points on circular arcs~\cite{CDR92,Bu98}.

\end{document}